\documentclass[pra,twocolumn,10pt]{revtex4}%
\usepackage{amsfonts}
\usepackage{amsmath}
\usepackage{amssymb}
\usepackage{graphicx}%
\newtheorem{theorem}{Theorem}

\newtheorem{corollary}[theorem]{Corollary}

\newtheorem{definition}[theorem]{Definition}

\newtheorem{lemma}[theorem]{Lemma}

\newenvironment{proof}[1][Proof]{\noindent\textbf{#1.} }{\ \rule{0.5em}{0.5em}}
\begin{document}
\title[Communication in XYZ All-to-All Quantum Networks with a Missing Link]{Communication in XYZ All-to-All Quantum Networks with a Missing Link}
\author{Sougato Bose}
\affiliation{Department of Physics and Astronomy, University College London, London WC1E
6BT, United Kingdom.}
\author{Andrea Casaccino}
\affiliation{Computer Architecture Group Lab, Information Engineering Department,
University of Siena, 53100 Siena, Italy}
\author{Stefano Mancini}
\affiliation{Department of Physics, University of Camerino, I-62032 Camerino, Italy}
\author{Simone Severini}
\affiliation{Institute for Quantum Computing and Department of Combinatorics \&
Optimization, University of Waterloo, Waterloo N2L 3G1, ON Canada}
\date{\today}

\begin{abstract}
We explicitate the relation between Hamiltonians for networks of interacting
qubits in the $XYZ$ model and graph Laplacians. We then study evolution in
networks in which all sites can communicate with each other. These are modeled
by the \emph{complete graph} $K_{n}$ and called \emph{all-to-all networks}. It
turns out that $K_{n}$ does not exhibit \emph{perfect state transfer} (PST).
However, we prove that deleting an edge in $K_{n}$ allows PST between the two
non-adjacent sites, when $n$ is a multiple of four. An application is routing
a qubit over $n$ different sites, by switching off the link between the sites
that we wish to put in communication. Additionally, we observe that, in
certain cases, the unitary inducing evolution in $K_{n}$ is equivalent to the
Grover operator.

\end{abstract}
\maketitle

\section{Introduction}

Networks of interacting qubits are a generalization of spin chains. These are
of theoretical importance for the study of many-body quantum systems and could
constitute a good test ground for technologies spanning from quantum key
distribution in multi-user networks to various nano-scale devices.

Since the first works \cite{ch, sou, sub} (see also \cite{sou1}, for a
review), networks of interacting qubits are considered to be good candidates
for engineering perfect quantum channels and allow information transfer
between distant particles. In this perspective, such networks appear to be
useful for the implementation of data buses in quantum mechanical devices, in
particular because undergoing a free dynamics after an initial set-up. One of
the problems arising in the scenario is given by natural dispersion effects,
which determine a loss of information often proportional to the distance
(\emph{i.e.}, the number of spins) between communicating sites. Ways to
circumvent the issue are based on a local tuning of the couplings \cite{ch} or
protocols for \emph{lifting }the encoding of qubits into multiparticle states
\cite{bb, os} (see also \cite{ek} and the references therein).

Another approach is to study what network topologies guarantee a small decay
of fidelity, when the couplings are homogeneous and constant during the entire
evolution of the system. Results in this direction have provided examples of
the so-called \emph{perfect state transfer} (for short, \emph{PST}), in
relation to combinatorial properties of the graphs modeling the networks (a
list of references on this area is in \cite{be}). In the $XY$ model, when
considering a single excitation, it has been shown that PST depends
essentially on the eigensystem of the adjacency matrix of the graph, because
certain invariant eigenspaces of the total Hilbert space evolve independently.
In the cases analyzed so far, even if PST occurs between two specific sites of
a network, routing arbitrarily over the entire network still requires an
external controller.

Here we will deal with PST in certain networks of spin-half particles with the
$XYZ$ interaction. Specifically, we will describe the structure of the
Hamiltonian which governs the single excitation setting. This turns out to be
proportional to the (combinatorial) Laplacian of the graph modeling the
network. Thus, it may be worth to remark that Hamiltonian arising in the
\emph{XY} and the\emph{XYZ} interactions are naturally associated to the
adjacency matrix and the Laplacian respectively and that these objects are the
two most common matrix representations of graphs.

We will study quantum evolution in a network in which all sites can
communicate between each other (in both directions). Such a network is modeled
by the \emph{complete graph} on $n$ vertices $K_{n}$ and it is called
\emph{all-to-all network}. The number of links in $K_{n}$ is $n(n-1)/2$. It
turns out that $K_{n}$ does not exhibit PST. However, we will show that
deleting an edge in $K_{n}$ will allow PST between the two non-adjacent sites
in the obtained network, when $n$ is a multiple of $4$ (apart the trivial case
$n=2$). This is modeled by a graph usually denoted by $K_{n}^{-}$, an
\emph{all-to-all network with a missing link}. The phenomenon is due to
interference effects and it is counterintuitive. On the side, we remark that
in certain cases the unitary inducing evolution in $K_{n}$ is the Grover
matrix up to an overall phase.

An application is routing a qubit on $K_{n}^{-}$ (\emph{i.e.}, transferring
the qubit between any two vertices of $K_{n}^{-}$, arbitrarily). The protocol
involves an external agent, whose role is to switch OFF only the link between
the sites that we wish to put in communication. When the process involves all
$n$ sites, it requires exactly $2n-1$ ON/OFF switching operations. This does
not provide any direct advantage over considering the empty graph as a network
(this is trivially the graph with no edges), with a single edge $ij$ if
communication between site $i$ and site $j$ is desired. On the other hand,
dynamics on $K_{n}^{-}$ is far more complex and the same edge configuration is
potentially useful for other tasks beyond perfect communication of a qubit.

The structure of the paper is as follows. In Section 2 we give minimal
background information on the $XYZ$ model and the necessary mathematical
definitions. We will clarify the exact relation between the $XYZ$ Hamiltonian
and the Laplacian matrix (this relation was mentioned in \cite{facer}, but not
given explicitly). In Section $3$ we study evolution on $K_{n}$. Section $4$
is devoted to $K_{n}^{-}$. Wi will also mention a generalization in which we
delete an arbitrary number of vertex disjoint edges. A brief conclusion is
drawn in Section $5$.

\section{Set-up}

Let $G=(V,E)$ be a simple undirected graph (that is, without loops or parallel
edges), with set of vertices $V(G)$ and set of edges $E(G)$. We take
$V(G)=\{1,...,n\}$ and assume that $|E(G)|=m$. The \emph{degree} $d(i)$ of a
vertex $i$ is the number of edges incident with $i$. The \emph{adjacency
matrix} of $G$ is denoted by $A(G)$ and defined by $[A(G)]_{ij}=1$, if $ij\in
E(G)$; $0$, $[A(G)]_{ij}=0$ if $ij\notin E(G)$.

The adjacency matrix is a useful tool to describe a network of $n$ spin-half
quantum particles. The particles are usually attached to the vertices of $G$,
while the edges of $G$ represent their allowed couplings. If one considers the
$XYZ$ interaction model (isotropic Heisenberg model), then $\{i,j\}\in E(G)$
means that the particles $i$ and $j$ interact by the Hamiltonian
$[H(G)]_{ij}=\left(  X_{i}X_{j}+Y_{i}Y_{j}+Z_{i}Z_{j}\right)  $, where $X_{i}%
$, $Y_{i}$ and $Z_{i}$ are the Pauli operators of the $i$-th particle (here we
consider unit coupling constant). In extension of the $XY$ model, we decide to
call this the $XYZ$ model, as the $Z$ interaction has been added to the $XY$
Hamiltonian. Thus, the Hamiltonian of the whole network reads
\begin{equation}
H_{XYZ}(G)=\frac{1}{2}\sum_{i\neq j}[A(G)]_{ij}\left(  X_{i}X_{j}+Y_{i}%
Y_{j}+Z_{i}Z_{j}\right)  , \label{Hnet}%
\end{equation}
and it acts on the Hilbert space $\mathcal{K}=\left(  \mathbb{C}^{2}\right)
^{\otimes n}$.

Let us now restrict our attention to the single excitation subspace
$\mathcal{H}\cong\mathbb{C}^{n}$, \emph{i.e.}, the subspace of dimension $n$
spanned by the vectors $\{|1\rangle,\ldots,|n\rangle\}$. A vector $|j\rangle$
indicates the presence of the excitation on the $j$-th site and the absence on
all the others. This is equivalent to the following tensor product of the $Z$
eigenstates $|\underset{n}{\underbrace{0\ldots010\ldots0}}\rangle$, being $1$
in the $j$-th position. In the basis $\{|1\rangle,\ldots,|n\rangle\}$, the
Hamiltonian coming from Eq. (\ref{Hnet}) has the following entries
\begin{equation}
\lbrack H_{XY}(G)]_{ij}=2[A(G)]_{ij},\qquad i\neq j \label{hi}%
\end{equation}
and%
\begin{equation}
\lbrack H_{Z}(G)]_{ii}=\frac{1}{2}\sum_{i,j}[A(G)]_{ij}-2\sum_{j}[A(G)]_{ij}
\label{hi1}%
\end{equation}
Eq. (\ref{hi}) comes from the term $X_{i}X_{j}+Y_{i}Y_{j}$ and Eq. (\ref{hi1})
comes from the term $Z_{i}Z_{j}$. It is then clear that in the $XY$ model the
relation between the Hamiltonian and adjacency matrix simply reduces to
$[H_{XY}(G)]_{ij}=2[A(G)]_{ij}$. Here, $H_{XYZ}(G)=H_{XY}(G)+H_{Z}(G)$. By
associating the vertex $i\in V(G)$ to the vector $|i\rangle\in\{|1\rangle
,...,|n\rangle\}$, we can introduce the following modified version of the
adjacency matrix:

\begin{definition}
\label{def1}The \emph{XYZ} \emph{adjacency matrix} of a graph $G$ is denoted
by $H(G)$ and defined by%
\[
\lbrack H_{XYZ}(G)]_{ij}=\left\{
\begin{tabular}
[c]{ll}%
$2,$ & if $ij\in E(G);$\\
$0,$ & if $ij\notin E(G);$\\
$m-2d(i),$ & if $i=j.$%
\end{tabular}
\ \ \right.
\]

\end{definition}

We shall drop the subscript $XYZ$ because in this paper we only deal the $XYZ$
model. The matrix $H(G)$ has a neat relation to the Laplacian of $G$. Let
$\Delta(G)$ be an $n\times n$ diagonal matrix such that $[\Delta
(G)]_{ii}=d(i)$. The (combinatorial) \emph{Laplacian }of $G$ is the matrix
$L(G):=\Delta(G)-A(G)$. Let $I_{n}$ be the $n\times n$ identity matrix.

\begin{lemma}
\label{le1}The XYZ adjacency matrix of a graph $G$ is $H(G)=mI_{n}-2L(G)$.
\end{lemma}

\begin{proof}
By Definition \ref{def1}, $H(G)=2A(G)+mI_{n}-2\Delta(G)$. By combining this
fact with the definition of Laplacian, we obtain the desired expression.
\end{proof}

\begin{corollary}
Given any graph $G$, the matrices $L(G)$ and $H(G)$ have a common set of
eigenvectors. Moreover, if $\mu$ is an eigenvalue of $L(G)$ then
$\lambda=m-2\mu$ is an eigenvalue of $H(G)$.
\end{corollary}

Let $\{|1\rangle,...,|n\rangle\}$ be the standard basis of an Hilbert space
$\mathcal{H}\cong\mathbb{C}^{n}$. We associate the vertex $i\in V(G)$ to the
vector $|i\rangle\in\{|1\rangle,...,|n\rangle\}$. Let $\iota=\sqrt{-1}$. Given
a graph $G$, two of its vertices $i$ and $j$, and a real number $0<t<\infty$,
the \emph{fidelity} at time $t$ between vertex $i$ and vertex $j$ is defined
by $f_{G}(i,j,t):=|\langle j|e^{-\iota H(G)t}|i\rangle|$. \ We say that two
vertices $i$ and $j$ of a graph $G$ admits \emph{perfect state transfer}
(\emph{w.r.t.} the XYZ adjacency matrix) if there exists $t$ for which
$f_{G}(i,j,t)=1$. Let $U_{t}(G)=e^{-\iota H(G)t}$ be the unitary matrix
associated to $H(G)$ as a function of $t$. Therefore, we have $f_{G}%
(i,j,t)=1$, if $\left\vert [U_{t}(G)]_{i,j}\right\vert =1$.

\section{All-to-All networks}

In this section, we consider a network in which all sites interact directly.
This network is modeled by the \emph{complete graph} $K_{n}$ on $n$ vertices.
In $K_{n}$, we have $d(i)=n-1$, for every $i\in V(G)$. The principal result of
the section is the following statement:

\begin{theorem}
\label{th1}Let $K_{n}$ be the complete graph on $n$ vertices. For every vertex
$i\in V(K_{n})$ and value $n\in\mathbb{N}$, we have the following:

\begin{itemize}
\item $\min_{t}f_{K_{n}}(i,i,t)=1-2/n$ and $t=\frac{\pi}{2n}+\frac{\pi k}{n}$,
for $k\geq0$;

\item $\max_{t}f_{K_{n}}(i,i,t)=1$ and $t=\pi k/n$, for $k\geq1$.
\end{itemize}

For every two distinct vertices $i,j\in V(K_{n})$ and value $n\in\mathbb{N}$,
we have the following:

\begin{itemize}
\item $\min_{t}f_{K_{n}}(i,j,t)=0$ is attained when $f_{K_{n}}(i,i,t)$ is maximum;

\item $\max_{t}f_{K_{n}}(i,j,t)=2/n$ is attained when $f_{K_{n}}(i,i,t)$ is minimum.
\end{itemize}
\end{theorem}

\begin{proof}
Since $d(i)=n-1$, for every $i\in V(K_{n})$, we have $A(K_{n})=J_{n}-I_{n}$,
where $J_{n}$ is the $n\times n$ all-ones matrix. Then, $L(K_{n}%
)=(n-1)I_{n}-(J_{n}-I_{n})=nI_{n}-J_{n}$. Given $|E(K_{n})|=\binom{n}{2}$, it
follows that
\[
H(K_{n})=\binom{n}{2}I_{n}-2\left(  nI_{n}-J_{n}\right)  =\frac{n\left(
n-5\right)  }{2}I_{n}+2J_{n}.
\]
The eigenvalues of $H(K_{n})$ are then $\lambda_{1}^{[1]}=n(n-1)/2$ and
$\lambda_{2}^{[n-1]}([H(K_{n})]_{ii}-[H(K_{n})]_{ij})=n\left(  n-5\right)
/2$. Let us denote by $\overrightarrow{J}_{n}$ the all-ones vector of
dimension $n$. Since $[H(K_{n}),J_{n}]=0$, these two matrices share a common
set of eigenvectors. The matrix $J_{n}$ is diagonalized by $F_{n}$, the
Fourier transform over the group $\mathbb{Z}_{n}$: $[F_{n}]_{ij}=e^{\iota2\pi
ij/n}\equiv\omega^{ij}$. As a consequence, the single eigenvector
corresponding to the eigenvalue $\lambda_{1}$ is of the form $|\lambda
_{1}\rangle=n^{-1/2}\overrightarrow{J}_{n}$; the $n-1$ eigenvectors
corresponding to $\lambda_{2}$ are of the form $|\lambda_{2}^{i}%
\rangle:=n^{-1/2}\sum_{j=1}^{n}\omega^{ij}|j\rangle$, for every $i=1,...,n-1$.
Now, we can read $U_{t}(K_{n})\equiv e^{-\iota H(K_{n})t}$ in its spectral
decomposition:
\[
U_{t}(K_{n})=n^{-1}\left(  e^{-\iota\lambda_{1}t}J_{n}+e^{-\iota\lambda_{2}%
t}\sum_{i;j,k=1}^{n-1;n}\omega^{ij}\overline{\omega}^{ik}|j\rangle\langle
k|\right)  .
\]
Then
\[
\lbrack U_{t}(K_{n})]_{ii}=n^{-1}e^{-\iota\binom{n}{2}t}+\frac{n-1}%
{n}e^{-\iota\frac{n\left(  n-5\right)  }{2}t}%
\]
for every $i$ and
\[
\lbrack U_{t}(K_{n})]_{ij}=n^{-1}\left(  e^{-\iota\binom{n}{2}t}%
-e^{-\iota\frac{n\left(  n-5\right)  }{2}t}\right)
\]
if $i\neq j$. Thus, we have $1-2/n\leq\left\vert \lbrack U_{t}(K_{n}%
)]_{ii}\right\vert \leq1$, for every $i$. The minimum is attained when
$t=\frac{\pi k}{n}+\frac{\pi}{2n}$. The maximum is attained for each $t=\pi
k/n$. So, $0\leq\left\vert \lbrack U_{t}(K_{n})]_{ij}\right\vert \leq2/n$. For
these entries, the minimum value is attained for each $t$ such that
$\left\vert [U_{t}(K_{n})]_{ii}\right\vert $ is maximum and \emph{viz}.
\end{proof}

\bigskip

The results is that there is no PST in $K_{n}$ (but for the trivial cases
$n\neq1,2$). This is equivalent to say that a qubit can not be accurately
transferred between two sites in a network in which all sites are connected to
each other. The fidelity given by the evolution in $K_{n}$, with $n=4,8$ is
illustrated in the figures below (left and right, \emph{resp.}) for
$t\in\lbrack0,\pi]$. The dashed lines represent $\left\vert [U_{t}%
(K_{n})]_{ij}\right\vert $ with $i\neq j$; the solid ones, $\left\vert
[U_{t}(K_{n})]_{ii}\right\vert $.
\[%
\begin{tabular}
[c]{|ll|}\hline%
{\includegraphics[
height=1.0573in,
width=1.6418in
]%
{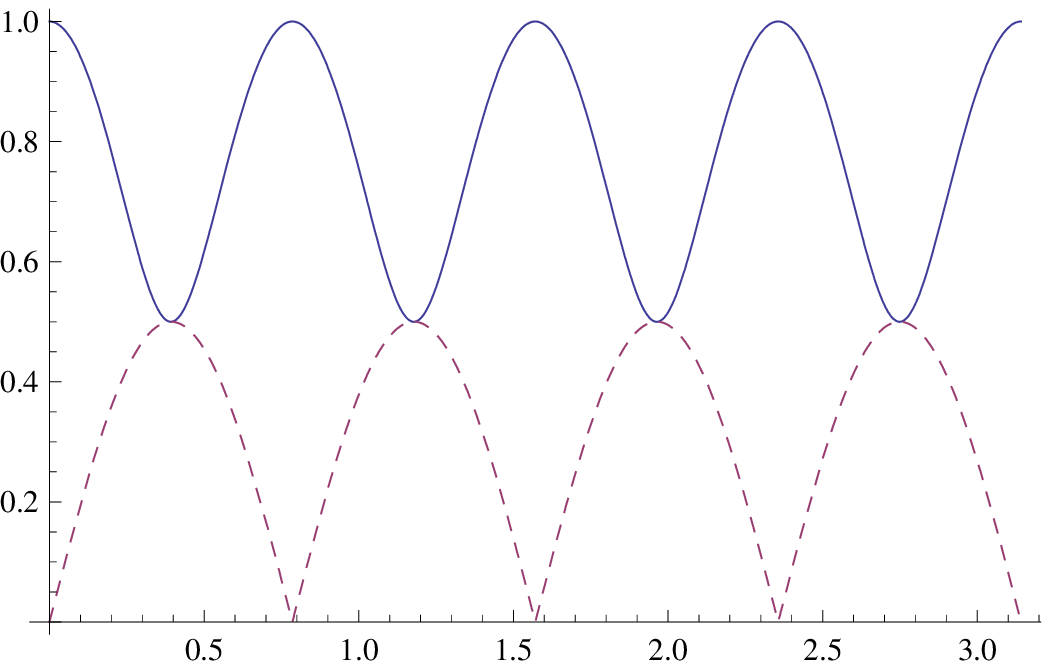}%
}
&
{\includegraphics[
height=1.0573in,
width=1.6418in
]%
{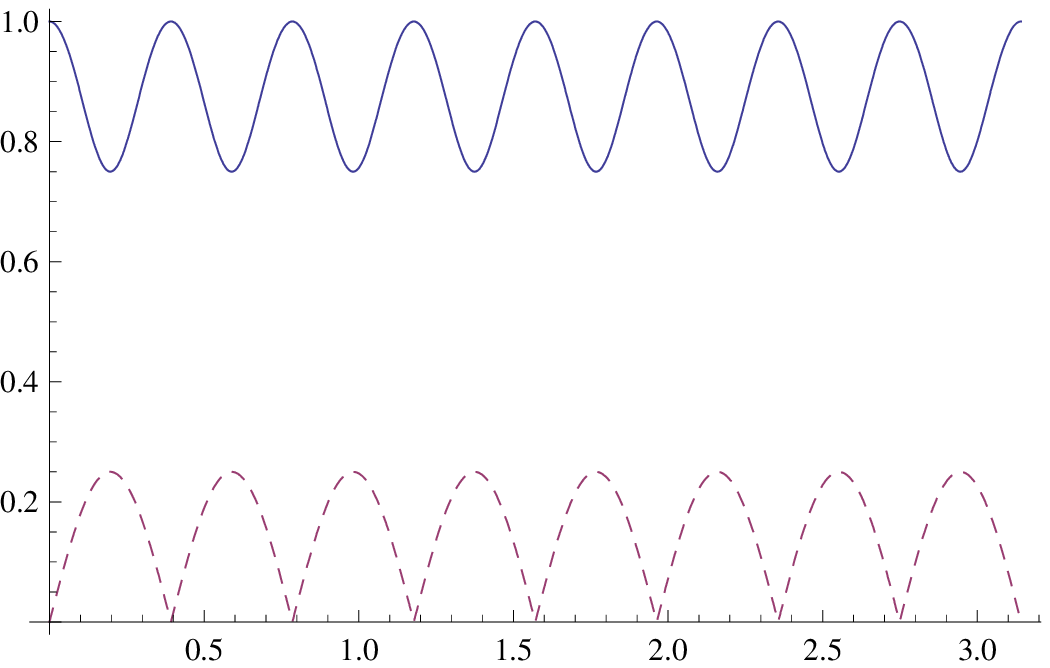}%
}
\\\hline
\end{tabular}
\]

For $t=\frac{\pi k}{n}+\frac{\pi}{2n}$, Theorem \ref{th1} prompts to the
following facts. If $n$ is odd then%

\[
\lbrack U_{t}(K_{n})]_{ij}=\left\{
\begin{tabular}
[c]{ll}%
$\delta_{ij}-2\iota/n,$ & if $n=4k-1;$\\
$\delta_{ij}-2/n,$ & if $n=4k+1.$%
\end{tabular}
\right.
\]
If $n$ is even, we need to distinguish two cases:

\begin{itemize}
\item $n=4k$,%
\[
\lbrack U_{t}(K_{n})]_{ij}=\left\{
\begin{tabular}
[c]{ll}%
$-\left(  \delta_{ij}-2/n\right)  \sqrt{i},$ & if $k$ even;\\
$\left(  \delta_{ij}-2/n\right)  \sqrt{i},$ & if $k$ odd;
\end{tabular}
\right.
\]

\item $n=4k+2$,%
\[
\lbrack U_{t}(K_{n})]_{ij}=\left\{
\begin{tabular}
[c]{ll}%
$\overline{-\left(  \delta_{ij}-2/n\right)  \sqrt{i}},$ & if $k$ even$;$\\
$\overline{\left(  \delta_{ij}-2/n\right)  \sqrt{i}},$ & if $k$ odd.
\end{tabular}
\ \ \right.
\]

\end{itemize}

This is the Grover operator with an overall phase. Recall that this operator
plays a central role in the Grover algorithm for database search. The unitary
$U_{t}(K_{n})$ induces a continuous-time quantum walk on $K_{n}$. Notably an
alternative version of this search technique has been designed as a
discrete-time quantum walk algorithm \cite{ma}.

\section{All-to-All networks with a missing link}

In this section, we consider a network in which all sites but two interact
directly. This network is modeled by $K_{n}^{-}$, the \emph{complete graph
minus an edge} on $n$ vertices. The graph $K_{n}^{-}$ is obtained from $K_{n}$
by deleting an arbitrary edge.

\begin{theorem}
\label{th2}Let $K_{n}^{-}$ be the complete graph minus an edge on $n$
vertices. For every vertex $i\in V(K_{n}^{-})$ and value $n=4k$, with
$k\in\mathbb{N}$, we have the following statements:

If $i=1,n$ then

\begin{itemize}
\item $\min_{t}f_{K_{n}^{-}}(i,i,t)=0$ and $t=\frac{\pi}{4}+\frac{\pi k}{2}$,
for $k\geq0$;

\item $\max_{t}f_{K_{n}^{-}}(i,i,t)=1$ and $t=\pi k/2$, for $k\geq1$.
\end{itemize}

If $i\neq1,n$ then

\begin{itemize}
\item $\min_{t}f_{K_{n}^{-}}(i,i,t)=1-2/n$ and $t=\frac{\pi}{2n}+\frac{\pi
k}{n}$, for $k\geq0$;

\item $\max_{t}f_{K_{n}^{-}}(i,i,t)=1$ and $t=\pi k/n$, for $k\geq1$.
\end{itemize}

If $i=1$ and $j=n$ then

\begin{itemize}
\item $\min_{t}f_{K_{n}^{-}}(1,n,t)=0$ is attained when $f_{K_{n}}(1,1,t)$ is maximum;

\item $\max_{t}f_{K_{n}^{-}}(1,n,t)=1$ is attained when $f_{K_{n}}(1,1,t)$ is minimum.

In all other cases, when $i\neq j$ (and $i\neq1,n$)

\item $\min_{t}f_{K_{n}^{-}}(i,j,t)=0$ is attained when $f_{K_{n}}(i,i,t)$ is maximum;

\item $\max_{t}f_{K_{n}^{-}}(i,j,t)=2/n$ is attained when $f_{K_{n}}(i,i,t)$
is minimum.
\end{itemize}
\end{theorem}

\begin{proof}
Let us define the $n\times n$ matrix $P_{n}$ such that $[P_{n}]_{1,1}%
=[P_{n}]_{n,n}=1$, $[P_{n}]_{1,n}=[P_{n}]_{n,1}=-1$, and $[P_{n}]_{ij}=0$,
otherwise. The Laplacian of $K_{n}^{-}$ can be written as
\begin{align*}
L(K_{n}^{-})  &  =\left(
\begin{array}
[c]{ccccc}%
n-2 & -1 & \cdots & -1 & 0\\
-1 & n-1 & -1 & \cdots & -1\\
\vdots & -1 & \ddots & -1 & \vdots\\
-1 & \cdots & -1 & n-1 & -1\\
0 & -1 & \cdots & -1 & n-2
\end{array}
\right) \\
&  =L(K_{n})-P_{n}.
\end{align*}
The XYZ adjacency matrix of $K_{n}^{-}$ has then the form%
\begin{align*}
H(K_{n}^{-})  &  =\left(  \binom{n}{2}-1\right)  I_{n}-2L(K_{n}^{-})\\
&  =\left(  \binom{n}{2}-1\right)  I_{n}-2\left(  nI_{n}-J_{n}-P_{n}\right)  .
\end{align*}
This matrix elements are of the form $[H(K_{n}^{-})]_{1,n}=[H(K_{n}%
^{-})]_{n,1}=0$, $[H(K_{n}^{-})]_{i,i}=m-2d(i)$ and $[H(K_{n}^{-})]_{i,j}=2$.
The eigenvalues of $H(K_{n}^{-})$ are then $\lambda_{1}^{[1]}=n(n-1)/2-1$,
$\lambda_{2}^{[1]}=[H(K_{n}^{-})]_{1,1}=n(n-1)/2-2n+3$ and $\lambda
_{3}^{[n-2]}=[H(K_{n}^{-})]_{2,2}-[H(K_{n}^{-})]_{1,2}=n(n-5)/2-1$. We can
chose an orthonormal basis of eigenvectors such that, from the spectral
decomposition of the unitary matrix $U_{t}(K_{n}^{-})\equiv e^{-\iota
H(K_{n}^{-})t}$, we have the following diagonal entries:%
\[
\lbrack U_{t}(K_{n}^{-})]_{ii}=\left\{
\begin{tabular}
[c]{l}%
$n^{-1}\left(  \frac{n}{2}e^{-\iota\left[  \binom{n}{2}-2n+3\right]
t}+e^{-\iota\left[  \binom{n}{2}-1\right]  t}\right.  $\\
$+\left.  (\frac{n}{2}-1)e^{-\iota\left[  \frac{n(n-5)}{2}-1\right]
t}\right)  ,$\\
if $i=1,n$;\\
$n^{-1}\left(  e^{-\iota\left[  \binom{n}{2}-1\right]  t}+(n-1)e^{-\iota
\left[  \frac{n(n-5)}{2}-1\right]  t}\right)  ,$\\
otherwise.
\end{tabular}
\right.
\]

Let us first consider the minimum value of $f_{K_{n}^{-}}(i,i,t)$. In general,
we would need to distinguish two cases depending on the parity of $n$. Here we
take $n$ to be a multiple of $4$. This implies that $n$ is always even. A
simple calculation shows the next facts. It follows that $\min_{t}f_{K_{n}%
^{-}}(1,1,t)=0$ and $\max_{t}f_{K_{n}^{-}}(1,1,t)=1$ if $t=\frac{\pi}{4}%
+\frac{\pi k}{2}$ ($k\geq0$), and $t=\pi k/2$ ($k\geq1$), respectively. The
same holds for $i=n$. If $i\neq1,n$, it follows that $\min_{t}f_{K_{n}^{-}%
}(i,i,t)=1-2/n$ and $\max_{t}f_{K_{n}^{-}}(i,i,t)=1$ if $t=\frac{\pi}%
{2n}+\frac{\pi k}{n}$ and $t=\pi k/n$, respectively.

The off-diagonal entries of $U_{t}(K_{n}^{-})$ are%
\[
\lbrack U_{t}(K_{n}^{-})]_{ij}=\left\{
\begin{tabular}
[c]{l}%
$n^{-1}\left(  e^{-\iota\left[  \binom{n}{2}-1\right]  t}-\frac{n}{2}%
e^{-\iota\left[  \binom{n}{2}-2n+3\right]  t}\right.  $\\
$+\left.  (\frac{n}{2}-1)e^{-\iota\left[  \frac{n(n-5)}{2}-1\right]
t}\right)  ,$\\
if $i=1$ and $j=n$;\\
$n^{-1}\left(  e^{-\iota\lbrack\binom{n}{2}-1]t}-e^{-\iota\lbrack\frac
{n(n-5)}{2}-1]t}\right)  ,$\\
otherwise.
\end{tabular}
\right.
\]

Given this equation, for $i=1$ and $j=n$, we have $\min_{t}f_{K_{n}^{-}%
}(1,n,t)=0$ when $t=\pi k/2$ and $\max_{t}f_{K_{n}^{-}}(1,n,t)=1$ when
$t=\frac{\pi}{4}+\frac{\pi k}{2}$, \emph{i.e.}, $f_{K_{n}^{-}}(1,n,t)$ is
minimum when $f_{K_{n}^{-}}(1,1,t)$ is maximum and \emph{viz.} The last case
to describe is when $i\neq j$ and $i\neq1,j\neq n$: $\min_{t}f_{K_{n}^{-}%
}(i,j,t)=0$ and $\max_{t}f_{K_{n}^{-}}(i,j,t)=2/n$ if $t=\pi k/n$ ($k\geq0$)
and $t=\frac{\pi}{2n}+\frac{\pi k}{n}$ ($k\geq1$), respectively.
\end{proof}

\bigskip

The fidelity given by the evolution in $K_{n}^{-}$, with $n=4,8,16,40$ is
illustrated in the figures below (top L/R, bottom L/R, \emph{resp.}) for
$t\in\lbrack0,\pi]$. The thick lines represent $f_{K_{n}^{-}}(1,1,t)$; the
thin ones, $f_{K_{n}^{-}}(1,n,t)$. The dashed lines represent $f_{K_{n}^{-}%
}(2,2,t)$; the dotted ones, $f_{K_{n}^{-}}(2,3,t)$.
\[%
\begin{tabular}
[c]{|ll|}\hline%
{\includegraphics[
height=1.0573in,
width=1.6418in
]%
{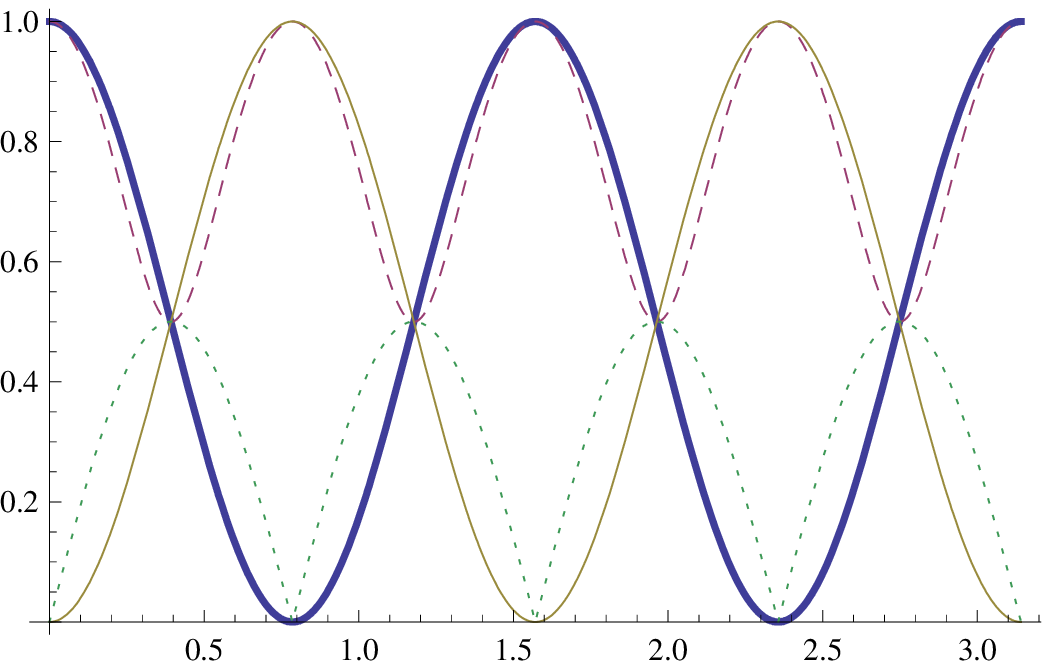}%
}
&
{\includegraphics[
height=1.0573in,
width=1.6418in
]%
{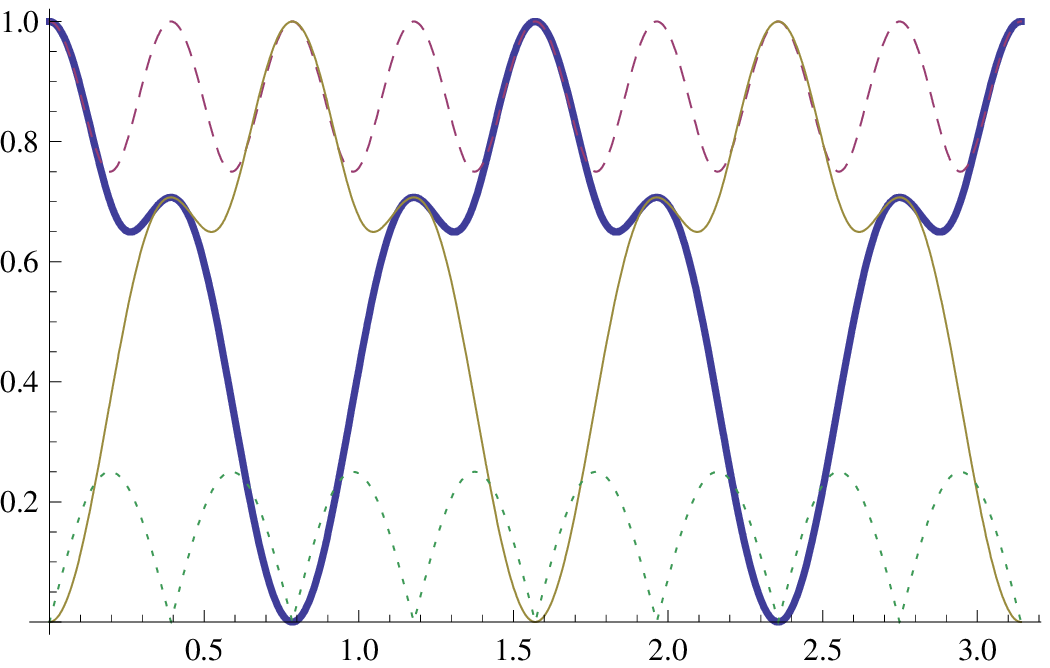}%
}
\\%
{\includegraphics[
height=1.0573in,
width=1.6418in
]%
{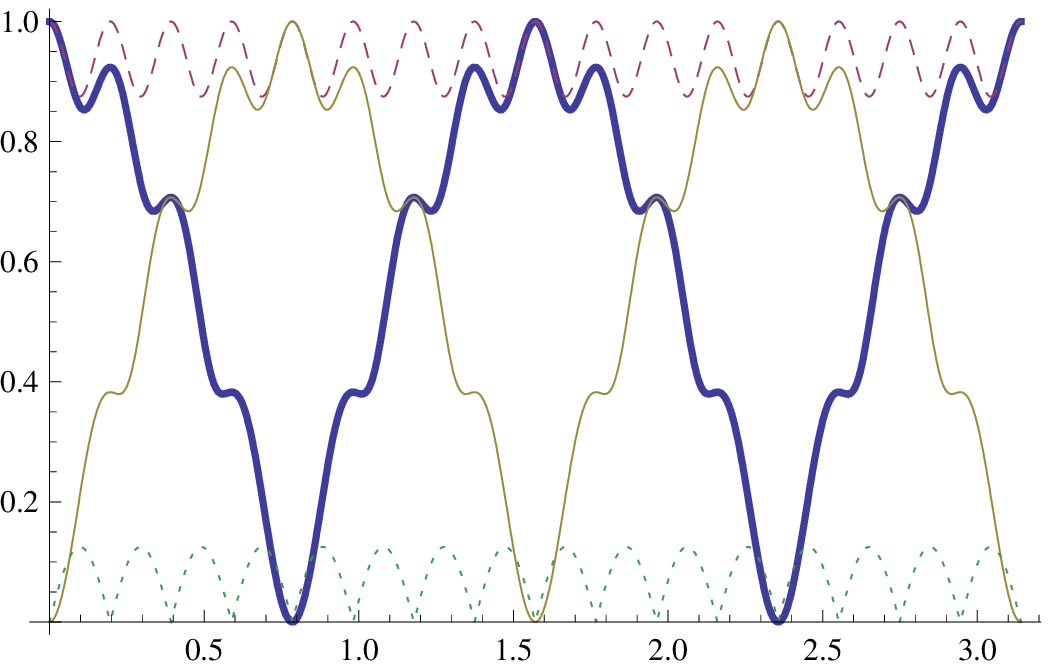}%
}
&
{\includegraphics[
height=1.0573in,
width=1.6418in
]%
{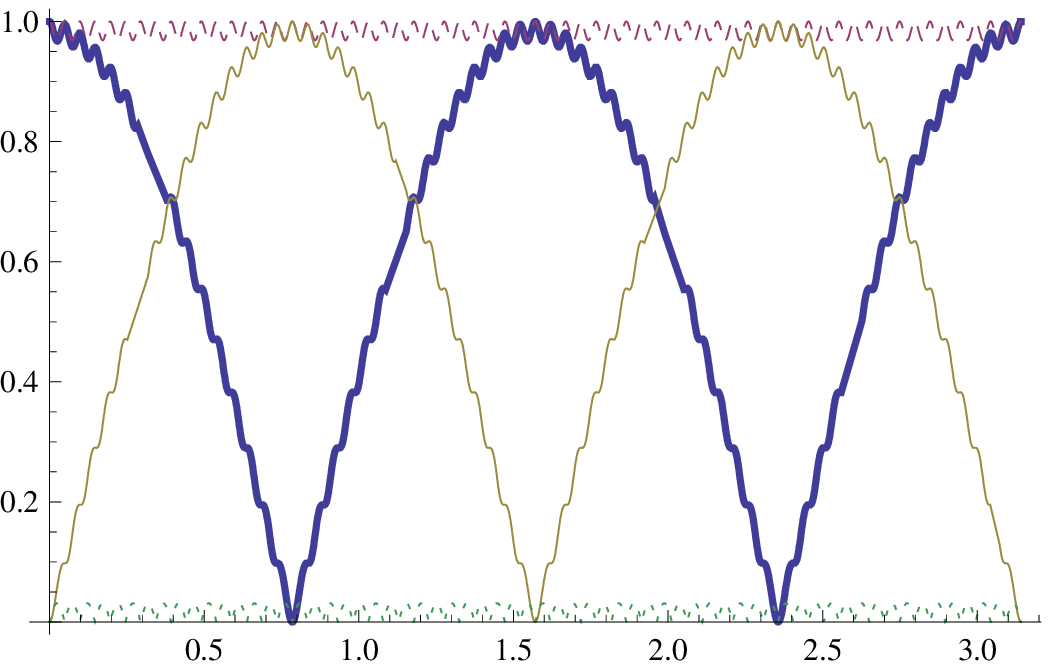}%
}
\\\hline
\end{tabular}
\ \
\]

When $n$ is a multiple of $4$, the unitary $U_{t}(K_{n}^{-})$ can be used to
route a qubit over $n$ sites. To put in communication site $i$ and site $j$,
we need to let the system evolve for a time $t=\frac{\pi}{4}+\frac{\pi k}{2}$,
once we have deleted from $K_{n}$ exactly the edge $ij$.

The setting described above can be generalized by deleting more than a single
edge from $K_{n}$. In fact, as far as we delete edges without common vertices
(or, in other words, \emph{vertex disjoint} edges), we still can obtain PST
between the end points of deleted edges, whenever the number of vertices in
the graph is a multiple of $4$. The maximum number of edges that can be
deleted is then $n/2$. In this case, the deleted set corresponds to a
\emph{perfect matching}, that is, a set of vertex disjoint edges including all
the vertices of the graph. As it is for $K_{n}^{-}$, PST occurs when
$t=\frac{\pi}{4}+\frac{\pi k}{2}$. Figure \ref{fig1} describes an example.%

\begin{figure}
[ptb]
\begin{center}
\includegraphics[
height=3.1169in,
width=3.2694in
]%
{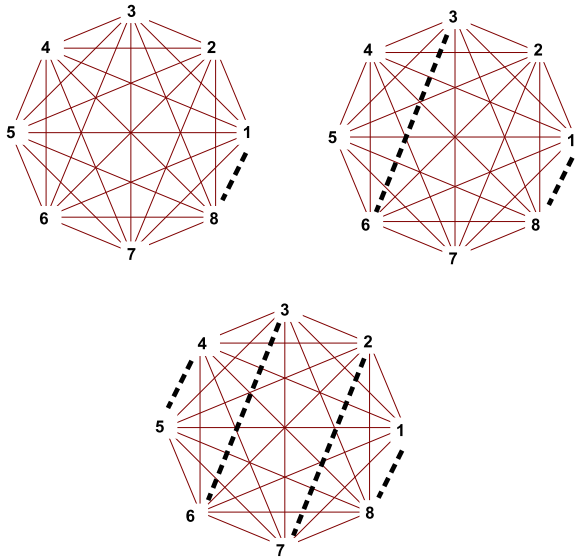}%
\caption{There is PST between vertex $1$ and $8$ in the top-left graph; for
the pairs $\{1,8\}$ and $\{3,6\}$ in the top-right one; for the pairs
$\{1,8\},\{2,7\},\{3,6\},\{4,5\}$ in the bottom graph. The dashed lines
represent deleted edges. }%
\label{fig1}%
\end{center}
\end{figure}

\section{Conclusion}

We have given an explicit relation between the Hamiltonian for networks of
interacting qubits in $XYZ$ model and the graph Laplacian. We have studied
evolution in networks modeled by the complete graph $K_{n}$ and the complete
graph minus one edge $K_{n}^{-}$. While $K_{n}$ does not allow PST, we have
seen that $K_{n}^{-}$ does allows PST between the two non-adjacent sites, when
$n$ is a multiple of four. This result can be used to arbitrarily route a
qubit between any two vertices of $K_{n}^{-}$.

\end{document}